\documentclass[sigconf,natbib=true,screen=true,anonymous=false,review=false]{acmart}

\author{Philipp Hager}
\orcid{0000-0001-5696-9732}
\affiliation{%
    \institution{University of Amsterdam}
    \city{Amsterdam}
    \country{The Netherlands}
}
\email{p.k.hager@uva.nl}

\author{Onno Zoeter}
\orcid{0000-0003-1704-706X}
\affiliation{%
    \institution{Booking.com}
    \city{Amsterdam}
    \country{The Netherlands}
}
\email{onno.zoeter@booking.com}

\author{Maarten de Rijke}
\orcid{0000-0002-1086-0202}
\affiliation{%
    \institution{University of Amsterdam}
    \city{Amsterdam}
    \country{The Netherlands}
}
\email{m.derijke@uva.nl}

\usepackage[inline]{enumitem}
\usepackage{multirow}
\usepackage{acronym}
\usepackage{stfloats}
\usepackage[skip=0pt]{caption}
\usepackage{etoolbox}
\usepackage{dsfont}

\makeatletter 
\newcommand\mynobreakpar{\par\nobreak\@afterheading} 
\makeatother

\usepackage[belowskip=-12pt,aboveskip=0pt]{caption}

\def\reals{\mathbb{R}}

\newcommand{\supp}[1]{\text{supp}\left( {#1} \right)}

\input{preamble/meta}

\begin{document}
\title[Unidentified and Confounded? Understanding Two-Tower Models for Unbiased Learning to Rank]{Unidentified and Confounded?\\Understanding Two-Tower Models for Unbiased Learning to Rank}



\begin{abstract}
    Additive two-tower models are popular learning-to-rank methods for handling biased user feedback in industry settings. Recent studies, however, report a concerning phenomenon: training two-tower models on clicks collected by well-performing production systems leads to decreased ranking performance. This paper investigates two recent explanations for this observation: confounding effects from logging policies and model identifiability issues. We theoretically analyze the identifiability conditions of two-tower models, showing that either document swaps across positions or overlapping feature distributions are required to recover model parameters from clicks. We also investigate the effect of logging policies on two-tower models, finding that they introduce no bias when models perfectly capture user behavior. However, logging policies can amplify biases when models imperfectly capture user behavior, particularly when prediction errors correlate with document placement across positions. We propose a sample weighting technique to mitigate these effects and provide actionable insights for researchers and practitioners using two-tower models.
\end{abstract}

\maketitle

\vspace*{-1mm}
\section{Introduction}

\begin{figure}
    \hspace*{-2em}
    \includegraphics[clip, trim=0mm -18mm 0mm 0mm, width=0.9\linewidth]{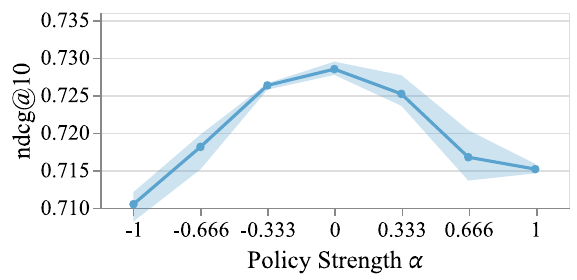}
    \vspace*{-1.4cm}
    \caption{Two-tower models trained on deterministic logging policies of varying strengths ($\alpha$) on MSLR30K: $\alpha = 1$ represents sorting by expert annotations, $\alpha = 0$ random sorting, and $\alpha = -1$ inversely ranking from least to most relevant.}
    \label{fig:example}
\end{figure}

Unbiased learning-to-rank (ULTR) addresses the challenge of optimizing ranking models for search and recommendation using implicit user feedback~\cite{Joachims2017IPW,Wang2016IPW,Gupta2024ULTR}. Clicks, in particular, are widely used for training models, yet they often provide a biased signal of user preferences and should not be naively interpreted as positive or negative preference~\cite{Joachims2005EyeTracking,Craswell2008PBM}. In web search, for example, clicks can be influenced by a document's position on the result page~\cite{Joachims2005EyeTracking}, its surrounding items~\cite{Zhuang2021XPA,Craswell2008PBM,Wu2021Feeds}, or the user's trust in the search engine to list relevant items first~\cite{Agarwal2019TrustBias,Joachims2005EyeTracking,Vardasbi2020Affine}.

Additive two-tower models have emerged as a popular solution for addressing click biases in industrial settings due to their flexibility and scalability~\cite{Guo2019PAL,Yan2022TwoTowers,Zhao2019AdditiveTowers,Haldar2020Airbnb,Khrylchenko2023Yandex}. These models consist of two neural networks, referred to as towers: the first network, or relevance tower, predicts item relevance based on content-related features. The second network, or bias tower, models user examination biases from contextual features such as document position, device type (e.g., mobile vs. desktop), or display height on the screen~\cite{Yan2022TwoTowers,Chen2023-TENCENT-ULTR-1}. During training, outputs from both towers are combined to predict user clicks, whereas, at serving time, only the relevance tower is employed. Thus, the bias tower effectively acts as a mechanism for debiasing the relevance tower during training.

Recent studies describe an interesting phenomenon when using two-tower models. As production logging policies become more effective at surfacing relevant items, the ranking performance of two-tower models trained on clicks collected by those policies steadily declines~\citep{Zhang2023Disentangling}. This is a troubling finding because most companies seek to optimize their production systems for ever-better ranking performance. 
Two key concepts have been put forward as contributing factors: \emph{confounding} and \emph{identifiability}.
\citet{Zhang2023Disentangling} theorize that stronger logging policies placing more relevant items in top positions act as a confounder by introducing a correlation between bias features and document relevance that two-tower models cannot disentangle. Likewise, position bias estimations of two-tower models might be biased on data collected by strong logging policies~\citep{Luo2024UnbiasedPropensity}. Although an explanation in terms of confounding is intuitive, it is unsatisfactory as two-tower models take both relevance-related document features and position into account.\footnote{Two-tower models condition on relevance features and position, essentially already blocking the problematic paths highlighted in \cite[Fig. 2b]{Zhang2023Disentangling} and \cite[Fig. 1b]{Luo2024UnbiasedPropensity}.} Yet, the empirical results seem consistent, and we observe the same effect in preliminary reproducibility experiments. Fig.~\ref{fig:example} displays the ranking performance of two-tower models trained on data collected by policies of varying strengths;\footnote{Different strengths in Fig.~\ref{fig:example} are based on sorting by expert relevance annotations with additive noise~\cite{Zhang2023Disentangling,Deffayet2023CMIP}. Section~\ref{sec:expert_policy} shows that this policy simulation is problematic.} ranking performance decreases when the logging policy is not uniformly random, indicating the existence of an underlying mechanism.

\citet{Chen2024Identifiability} discuss an alternative explanation in terms of identifiability. Intuitively, identifiability asks under which conditions we can (uniquely) recover model parameters from observed data. Inferring model parameters in click modeling is a well-known challenge~\cite{Chuklin2015ClickModels,Joachims2017IPW}. E.g., estimating position bias parameters typically requires observations of the same query-document pair across two positions (a swap) to correctly attribute changes in clicks to bias rather than a shift in document relevance~\cite{Craswell2008PBM,Agarwal2019AllPairs,Fang2019InterventionHarvesting}. \citet{Chen2024Identifiability} analyze click models following the Examination Hypothesis (the assumption that clicks factorize into bias and relevance terms, which two-tower models make~\cite{Yan2022TwoTowers}) and show that we can only recover click model parameters on data in which all positions are \emph{connected} through document swaps. The influence of logging policies on identifiability seems clear: our production system might not collect suitable data for training two-tower models, e.g., by never swapping documents across positions~\cite{Chen2024Identifiability}. The experiments on logging policy confounding in \cite{Zhang2023Disentangling,Luo2024UnbiasedPropensity} and Fig.~\ref{fig:example} use deterministic logging policies without document swaps, suggesting identifiability is a culprit. However, work on identifiability suggests rather binary observations, i.e., sufficient swaps connect all positions or do not.~\cite{Chen2024Identifiability}. Hence, can identifiability explain a gradually deteriorating ranking performance? And is there still an effect of logging policies on two-tower models if identifiability is guaranteed?

Given their widespread usage in industry and academia, it is important to understand precisely when two-tower models work and under what conditions they might fail. Thus, our work tackles the following three research questions:

\vspace{1mm}
\begin{enumerate}[nosep, leftmargin=*, label = \textbf{(RQ\arabic*)}]
    \item What are the conditions for identifying two-tower models?
    \item Is there an impact of logging policies on two-tower models beyond identifiability?
    \item What are strategies to alleviate potential logging policy influences?
\end{enumerate}

\noindent
We first show that two-tower models require either randomized document swaps or overlapping feature distributions across ranks to be identified (RQ1). We then show that while logging policies do not influence well-specified two-tower models, they can amplify biases in misspecified models when prediction errors correlate with document placement (RQ2). Finally, we propose a sample weighting technique that helps to alleviate the bias amplification introduced by logging policies (RQ3). Our contributions are:

\vspace{1mm}
\begin{itemize}[nosep, leftmargin=*]
    \item A theoretical analysis disentangling the problem of identification and logging policy effects for additive two-tower models.
    \item A propensity weighting scheme helping to alleviate logging policy effects on misspecified two-tower models.
    \item A thorough empirical evaluation of all theoretical contributions through simulation experiments.
    \item A discussion of practical takeaways for researchers and industry practitioners working with two-tower models.
\end{itemize}

\noindent
All code, data, and complete simulation results are available at:\\\url{https://github.com/philipphager/two-tower-confounding}.

\vspace*{-3mm}
\section{Related Work}

\paragraph{Two-tower models for unbiased learning-to-rank}
\citet{Guo2019PAL} introduce the industry practice of using two-tower models to account for click biases to the research community. \citet{Yan2022TwoTowers} closely examine and question the additive assumption in two-tower models and suggest alternatives that account for mixtures of user behaviors. \citet{Zhao2019AdditiveTowers} train two-tower models with multiple types of feedback beyond clicks for video recommendation at YouTube. \citet{Haldar2020Airbnb} and \citet{Khrylchenko2023Yandex} use two-tower models for search at Airbnb and Yandex respectively. \citet{Zhuang2021XPA} and \citet{Wu2021Feeds} include surrounding items and user context when modeling bias towers. \citet{Hager2023ClickModelIPS} investigate parallels between two-tower models and pointwise inverse propensity scoring~\cite{Joachims2017IPW,Vardasbi2020Affine,Oosterhuis2020PolicyAware}. And \citet{Chen2023-TENCENT-ULTR-1} show the importance of considering bias features beyond position in web search.

\vspace*{-3mm}
\paragraph{Logging policy confounding}
The idea that logging policy performance affects two-tower models and click models more broadly has recently gained attention in unbiased learning-to-rank~\cite{Deffayet2023Robustness,Deffayet2023CMIP,Zhang2023Disentangling,Luo2024UnbiasedPropensity}. \citet{Deffayet2023Robustness} show that a strong ranking performance of click models does not guarantee inferring unbiased parameters when training on strong logging policies. \citet{Zhang2023Disentangling} introduce logging policy confounding for additive two-tower models, arguing that policies performing better than a random ranker break the independence assumption between bias and relevance towers. They use deterministic logging policies derived from expert annotations, which likely leads to identifiability problems (cf.\ Section~\ref{sec:identifiability}) and their logging policy simulation likely introduces new confounding problems (cf.\ Section~\ref{sec:expert_policy}). We build on their findings, separating logging policy effects from identifiability problems, confirming their empirical results, but showing that the primary problems are model misspecification and a flawed simulation setup. \citet{Luo2024UnbiasedPropensity} find that positions containing more relevant documents systematically inflate examination probabilities. Empirically, we confirm their findings under certain conditions and propose a simple sample-weighting scheme for binary cross-entropy to adjust for unequal item exposure under a logging policy reminiscent of corrections used in intervention harvesting for position bias estimation~\cite{Agarwal2019AllPairs,Fang2019InterventionHarvesting,Benedetto2023ContextualBias}.

\vspace*{-3mm}
\paragraph{Identifiability} The study of identifiability, fundamentally the question about when we can recover model parameters from observational data, dates back decades~\cite{Lewbel2019IdentificationZoo,Koopmans1950Identification,Hurwicz1950Identification}, but only recently emerged in unbiased learning-to-rank. \citet{Oosterhuis2022Limitations} demonstrates that a lack of document swaps across positions leads to cases in which infinitely many model parameters explain observed clicks equally well, leading to inconsistent click models. \citet{Chen2024Identifiability} connected this example to formal identifiability theory, finding that click models following the Examination Hypothesis~\cite{Richardson2007PBM,Craswell2008PBM} require randomized swaps creating a connected position graph. Our work extends this analysis to additive two-tower models and shows that identifiability can be achieved not only through exact document swaps but also through overlapping feature distributions. While \citet{Chen2024Identifiability} noted logging policies might affect identifiability, we demonstrate they can impact two-tower models in ways that extend beyond identifiability issues in practical settings.

\vspace*{-1mm}
\section{Two-Tower Models}

Our analysis focuses on additive two-tower models, the most prevalent type in practical applications~\cite{Guo2019PAL,Yan2022TwoTowers,Zhao2019AdditiveTowers,Haldar2020Airbnb}. Each tower processes different inputs: one tower handles relevance-related query-document features, while the other uses context-dependent bias features. We will focus the majority of our analysis on position bias~\cite{Joachims2005EyeTracking,Craswell2008PBM}. Each tower outputs logits\footnote{Logits are the natural logarithm of the odds of a probability $p$: $\text{logit}(p) = \ln(\text{odds}(p))$} which are added together and transformed to a click probability via the sigmoid function:
\begin{equation}
    \label{eq:two-towers}
    P(C = 1 \mid q, d, k) = \sigma(\theta_{k} + \gamma_{q,d}), \text{with: } \sigma(x) = (1 + e^{-x})^{-1}.
\end{equation}
where the probability of click $C$ is a combination of a position bias logit $\theta_{k} \in \reals$ for rank $k \in K, K = \{1, 2, \dots \}$, and a relevance logit $\gamma_{q,d} \in \reals$ for document $d \in D$ and query $q \in Q$. We use $\sigma(\cdot)$ to denote the sigmoid function. The most common way to train this model is by minimizing the expected negative log-likelihood:
\begin{equation}\label{eq:expected_log_likelihood}
    \begin{split}
        \mathcal{L}(\theta, \gamma) = - \mathbb{E}_{(q,d,k,c) \sim P(Q,D,K,C)} \bigl[c\log \sigma(\theta_k + \gamma_{q,d}) \\
+ (1-c)\log (1-\sigma(\theta_k + \gamma_{q,d}))\bigr]
    \end{split}
\end{equation}
As the true data-generating distribution $P(Q, D, K, C)$ is unknown, we typically estimate model parameters by minimizing the empirical negative log-likelihood~\cite{Yan2022TwoTowers}:
\begin{equation}
    \begin{split}
        \mathcal{\hat{L}}(\theta, \gamma: D) = - \frac{1}{N} \sum_{(q,d,k,c) \in \mathcal{D}} \bigl[c\log \sigma(\theta_k + \gamma_{q,d}) \\
        + (1-c)\log (1-\sigma(\theta_k + \gamma_{q,d}))\bigr],
    \end{split}
\end{equation}
over a dataset of $N$ observations: $\mathcal{D} = \{ \left( q_i, d_i, k_i, c_i \right) \}_{i=1}^N$.

\vspace*{-1mm}
\section{RQ1: Identifiability of Two-Tower Models}
\label{sec:identifiability}

We begin our analysis by examining when we can uniquely recover the parameters of two-tower models from click data, fundamentally a question of identifiability. We will define identifiability, analyze it for additive two-tower models similar to \citet{Chen2024Identifiability}, and extend this analysis to identifiability from query-document features.

\vspace*{-1mm}
\subsection{Identifiability}

We define identifiability following the standard work of \citet[Definition 19.4]{Koller2009PGM}, whereby a model is identifiable if no two distinct sets of parameters $(\theta_k, \gamma_{q,d}), (\theta_k', \gamma_{q,d}')$ induce the same distribution over observable variables. Meaning, if:
\begin{equation}
    P(C \mid q, d, k; \theta, \gamma) = P(C \mid q, d, k; \theta', \gamma')
\end{equation}
it must follow that $\theta_k = \theta_k'$ for all positions and $\gamma_{q,d} = \gamma_{q,d}'$ for all query-document pairs. This is an important foundational property in probabilistic modeling ensuring a one-to-one (injective) mapping between model parameters and the observable data distribution. Unidentifiability, in turn, implies that we cannot uniquely recover the true model parameters from observed data as there are multiple (possibly infinite) parameter combinations that explain the observed data equally well.

\vspace*{-1mm}
\subsection{Invariance to parameter shift}

To begin our analysis, we first must acknowledge that two-tower models face an inherent identifiability challenge: we only observe clicks, while position bias and document relevance remain unobserved. This leads to a fundamental problem of parameter shifts: we can arbitrarily increase our relevance parameters while decreasing our bias parameters by the same amount, resulting in identical click probabilities but different parameter values.

To demonstrate this shift, let's consider a set of alternative parameters $\theta_{k}'$ and $\gamma_{q,d}'$ next to our true model parameters $\theta_{k}$ and $\gamma_{q,d}$. Since the sigmoid function $\sigma(\cdot)$ is injective, identical click probabilities require identical inputs to the sigmoid. Therefore, any two parameter sets producing the same click probability must satisfy:
\begin{equation}\label{eq:shift}
    \theta_{k} + \gamma_{q,d} = \theta_{k}' + \gamma_{q,d}' \quad \forall \, (q,d,k)
\end{equation}
We can construct valid alternative parameters that follow Eq.~\ref{eq:shift} by considering a real-valued constant $\Delta_k \in \reals$ for any rank $k \in K$:
\begin{equation}\label{eq:deltas}
    \theta_{k}' = \theta_{k} + \Delta_k, \quad \gamma_{q,d}' = \gamma_{q,d} - \Delta_k.
\end{equation}
This shows the invariance of additive two-tower models to additive parameter shifts. By rearranging, we highlight:
\begin{equation}\label{eq:deltas-2}
    \gamma_{q,d} - \gamma_{q,d}' = \theta_{k}' - \theta_{k} = \Delta_k,
\end{equation}
that the left-hand side depends solely on query $q$ and document $d$, while the right-hand side depends only on the rank $k$. Consequently, when documents appear exclusively in a single position (never appearing in multiple ranks), we can choose an arbitrary offset $\Delta_k$ to shift their relevance parameters and adjust the bias parameter accordingly without changing the overall click probabilities. To illustrate, consider two documents $d_1$ and $d_2$  that are displayed exclusively at ranks 1 and 2 respectively:
\begin{equation}\label{eq:two-deltas}
    \gamma_{q,d_1} - \gamma_{q,d_1}' = \Delta_1, \quad
    \gamma_{q,d_2} - \gamma_{q,d_2}' = \Delta_2.
\end{equation}
Without observing one document across both positions, the offsets \(\Delta_1\) and \(\Delta_2\) can be chosen independently, preventing a unique decomposition of bias and relevance parameters. Therefore, \textit{additive two-tower models are unidentifiable when observing each query-document pair only at a single position~\cite{Chen2024Identifiability,Oosterhuis2022Limitations}}.

\vspace*{-1mm}
\subsection{Identification through randomization}
\label{sec:identifiability-randomization}

Now suppose that we observe a query-document pair $q,d$ across positions 1 and 2. Note that the left-hand side of Eq.~\ref{eq:deltas-2} depends only on the current document. Thus, observing the exact same query-document pair across two positions forces the parameter offsets between both positions to be equal:
\begin{equation}
    \gamma_{q,d} - \gamma_{q,d}' = \Delta_1 = \Delta_2.
\end{equation}
Note that we do not need all query-document pairs to swap across all positions. \citet{Chen2024Identifiability} observe that it is sufficient that an undirected graph $G = (V, E)$, where vertices $V$ correspond to positions and edges $E$ exist between each pair of ranks that shares at least one query-document pair, is connected~\cite[Theorem 1]{Chen2024Identifiability}. Similarly, for additive two-tower models, if a graph of positions is connected, all rank-dependent offsets $\Delta_k$ must be equal to a global offset $\Delta$\footnote{\citet{Chen2024Identifiability} consider multiplicative parameter shifts under the Examination Hypothesis~\cite{Richardson2007PBM,Craswell2008PBM}, but we can directly apply their graph reasoning to additive shifts.}. To readers familiar with literature on position bias estimation, it might be useful to think of different randomization schemes as different ways to connect a graph of positions~\cite{Radlinski2006FairPairs,Craswell2008PBM,Joachims2017IPW}. Note that when the graph contains disconnected components, each component can have its own arbitrary parameter shift, independent of other components, which fundamentally prevents the recovery of a single unique set of model parameters.

Thus, random swaps unify parameter shifts across positions to a single global parameter shift. This last remaining ambiguity is commonly resolved through normalizing parameters~\cite[Section 6.3]{Lewbel2019IdentificationZoo}. A standard normalization for additive parameter shifts are location normalizations, e.g., fixing the first bias parameter $\theta_1=0$ (used in this work) or enforcing bias parameters to be centered around zero: $\sum_{k = 1}^{K} \theta_{k} = 0$. By normalizing bias parameters, we can finally recover a unique set of both bias and relevance parameters. Thus, \textit{additive two-tower models are identifiable up to an additive constant when observing query-document pairs across positions, such that all positions form a connected graph.} 

\vspace*{-1mm}
\subsection{Identification through overlapping features}
\label{sec:Identification-feature-overlap}

In most applications of two-tower models, we do not assign independent model parameters to each query-document pair. Instead, we predict relevance from a shared feature representation. In this setting, every query-document pair is represented by a feature vector $x_{q,d} \in \reals^m$ and the relevance tower $r(\cdot)$ consists, for example, of a linear model or a neural network:
\begin{equation}
    P(C = 1 \mid q, d, k) = \sigma(\theta_{k} + r(x_{q,d})).
\end{equation}
Note that when discussing neural networks in this work, we ask if the output logits are identified (nonparametric identification \cite[Section 6]{Lewbel2019IdentificationZoo}) not the network's parameters. Our main result builds on \cite[Theorem 1]{Chen2024Identifiability} and shows that identifiability is possible without swaps as long as query-document features overlap across ranks:

\begin{theorem}[Identifiability through feature overlap]
    \label{theorem:feature-overlap}
    Let $G = (V, E)$ be an undirected graph where vertices $V$ correspond to positions and an edge $(k, k') \in E$ exists if the feature support\footnote{Colloquially, the set of feature vectors that can appear at rank $k$.} between positions $k$ and $k'$ overlap:
    \begin{equation}
        \label{eq:identifiability-features}
           \begin{split}
            \supp{P(x \mid k)} &\cap \supp{P(x \mid k')} \neq \emptyset.
           \end{split}
    \end{equation}
    If $G$ is connected and the relevance tower $r(\cdot)$ is continuous, then the additive two-tower model is approximately identifiable up to an additive constant.
    \end{theorem}
    
    \begin{proof}
    Consider the alternative parameterization $\theta'_k$ and $r'(\cdot)$ that yield identical click probabilities to our original parameters $\theta_k$ and $r(\cdot)$. Recalling Eq.~\ref{eq:deltas-2}, for these parameterizations to produce the same click probabilities, the following must hold:
    \begin{equation}\label{eq:feature_deltas}
        r(x_{q,d}) - r'(x_{q,d}) = \theta'_{k} - \theta_{k} = \Delta_k,
    \end{equation}
    where $\Delta_k$ is a rank-dependent offset. For any two positions $k$ and $k'$ that share overlapping feature support, there exist query-document pairs with feature vectors $x_1 \sim P(x|k)$ and $x_2 \sim P(x|k')$ such that $x_1 \approx x_2$. By the continuity of $r(\cdot)$ and $r'(\cdot)$, we have:
    \begin{equation}
        r(x_1) - r'(x_1) = \Delta_k \quad \text{and} \quad r(x_2) - r'(x_2) = \Delta_{k'}.
    \end{equation}
    Assuming that $r(\cdot)$ and $r'(\cdot)$ share a Lipschitz constant $L$, we can bound the difference in parameter offsets between positions as:
    \begin{align}
    |\Delta_k - \Delta_{k'}| \leq 2L \cdot \|x_1 - x_2\|_2.
    \end{align} 
    This bound shows that when feature vectors are similar between positions, their parameter offsets must also be similar. As the feature overlap increases (i.e., as $\|x_1 - x_2\|_2 \to 0$), the difference in offsets approaches zero. Given that the graph $G$ is connected, there exists a path between any two positions $k$ and $k'$. Along this path, the parameter offsets between adjacent positions are approximately equal due to the continuity constraint. By transitivity across the connected graph, all position-dependent offsets $\Delta_k$ must converge to a single global offset $\Delta$ up to an approximation error that diminishes as feature overlap increases. After normalization (e.g., setting $\theta_1 = 0$), the model parameters are uniquely identified.
\end{proof}

\noindent%
We showed that when query-document features are similar between positions, their parameter offsets must also be similar. Therefore, to identify a model from features, we need feature spaces that overlap between positions. The overlap condition ensures that we can find similar query-document features across ranks and the continuity condition constrains parameter offsets across positions to enable approximate model identification. Thus, \emph{when document swaps across ranks are unavailable, a continuous relevance model combined with overlap in the feature distributions across positions can (approximately) identify two-tower models up to an additive constant}.

In practice, many two-tower models use bias features beyond position, such as device platform, content type, or the display height on screen~\cite{Zou2022Baidu,Chen2023-TENCENT-ULTR-1,Yan2022TwoTowers}. Such information might be captured in a feature vector $z_{q,d} \in \reals^n$ serving as input to the bias tower $b(\cdot)$:
\begin{equation}
    P(C = 1 \mid q,d) = \sigma(b(z_{q,d}) + r(x_{q,d})).
\end{equation}
Our identifiability Theorem~\ref{theorem:feature-overlap} directly extends to this model setup, but with increased combinatorial complexity. The model remains identifiable only when sufficient overlap exists in document distributions across the entire space of bias feature combinations. We direct the interested reader on this topic to Appendix~\ref{appendix:beyond-position}.

\vspace*{-1mm}
\subsection{Practical pitfalls of overlapping features}

The previous section makes two key assumptions for identifiability that can be difficult in practice: overlapping document features and a continuous relevance model. First, feature overlap decreases with increasing dimensionality~\cite{DAmour2021Overlap}. That is, as we add more query-document features, documents become less likely to be sufficiently close in feature space. Therefore, we should aim to use fewer query-document features, use dimensionality reduction methods, or introduce document swaps to guarantee overlapping support.

Secondly, the continuity assumption requires that small differences in feature space do not cause large, discontinuous jumps in relevance predictions. While neural networks are continuous, deep networks can produce large jumps in relevance prediction even for minor feature differences. Thus, when randomized data is unavailable, it is advisable to use shallow neural networks, regularization, or making parametric assumptions when possible to limit the expressiveness of the relevance model. Note that this advice may conflict with our later discussion of model misspecification in Section~\ref{sec:logging_policy}.

\vspace*{-1mm}
\section{RQ2: Influence of the Logging Policy}

We have seen that identifying additive two-tower models requires overlap in distributions across ranks, either through explicit randomization or shared query-document features that allow the disentangling of bias and relevance. Therefore, it is important first to highlight that logging policies govern the query-document distributions across ranks and, thus, have a major impact on identifiability. E.g., a single, deterministic logging policy will never collect a dataset that is suitable to identify a two-tower model with separate relevance parameters per query-document pair, as this model requires explicit document swaps (Section~\ref{sec:identifiability-randomization}). Even when we generalize across query-document features, we require overlap in the feature distributions across ranks, which are fundamentally governed by the logging policy. Thus, \emph{logging policies may collect data that is insufficient for identifying two-tower models.}

\vspace*{-1mm}
\subsection{Influence beyond identifiability?}
The connection between a logging policy and identifiability seems clear: we must collect overlapping feature or document distributions across ranks. Our analysis so far, and prior work on identifiability~\cite{Chen2024Identifiability}, suggest that ensuring that a logging policy collects sufficient document swaps across all bias dimensions guarantees that two-tower models can recover model parameters given enough data. Recent work \cite{Zhang2023Disentangling,Luo2024UnbiasedPropensity} and our experiment in Fig.~\ref{fig:example}, however, found gradual degradation in ranking performance when training two-tower models on logging policies with gradually increasing ranking performance, hinting at an influence of the logging policy beyond a rather binary identifiability condition. Next, we investigate which role logging policies might play when training two-tower models.

\vspace*{-1mm}
\subsection{Logging policy impact on model estimation}
\label{sec:logging_policy}

For simplicity, we examine the role of the logging policy for a two-tower model with separate relevance parameters for each query-document pair. Our findings below translate to the case in which we generalize over query-document features. To begin, we highlight the role of the logging policy by rewriting the expected negative log-likelihood from Eq.~\ref{eq:expected_log_likelihood}:
\begin{equation}
    \begin{split}
        &\mathcal{L}(\theta, \gamma)=
        - \sum_{q} P(q)\! \sum_{d, k} \pi(d, k \mid q) \bigl[ P(C\!=\!1 \mid q,d,k) \cdot {}
        \\[-4pt]
        &\mbox{}\hspace*{0.5cm} 
        \log \sigma(\theta_k + \gamma_{q,d})+ P(C\!=\!0|q,d,k)\log (1-\sigma(\theta_k + \gamma_{q,d})) \bigr].
    \end{split}
\end{equation}
We define a policy as the joint probability $\pi(d, k \mid q)$ of a document $d$ being shown at rank $k$ for query $q$. We can observe that the policy weights the contribution of each query-document pair to the loss. The analysis below separates the impact of logging policies on well-specified versus misspecified models.

\begin{lemma}[No policy impact on well-specified models]
\label{lemma:well-specified}
If a two-tower model is well-specified (i.e., can perfectly model the true click probabilities) and identifiable, then the logging policy has no effect on the estimated model parameters.
\end{lemma}

\begin{proof}
The partial derivatives of the loss function with respect to the model parameters are:
\begin{equation}
\begin{split}
    &\mbox{}\hspace*{-3mm}\frac{\partial \mathcal{L}}{\partial \gamma_{q,d}} \!=\!{}\\[-3pt]
    &\mbox{}\hspace*{5mm}P(q) \! \sum_{k} \pi(d, k \mid q) \bigl[ P(C\!=\!1 \mid q,d,k) - \sigma(\theta_k + \gamma_{q,d}) \bigr]\!=\!0,
    \hspace*{-3mm}\mbox{}
\end{split}
\label{eq:relevance-derivative}
\end{equation}
\vspace*{-3mm}
\begin{equation}
\begin{split}
    & \frac{\partial \mathcal{L}}{\partial \theta_k} \!=\! {}\\
    & \sum_{q} P(q) \! \sum_{d} \pi(d, k \mid q) \bigl[ P(C\!=\!1 \mid q,d,k) - \sigma(\theta_k + \gamma_{q,d}) \bigr]\!=\!0.
\end{split}    
\label{eq:bias-derivatie}
\end{equation}
%
%
For these gradients to vanish at the optimal parameters, the following condition must hold for every query-document pair with non-zero display probability, i.e., when $\pi(d, k \mid q) > 0$:
\begin{equation}
    P(C\!=\!1 \mid q,d,k) = \sigma(\theta_k + \gamma_{q,d}).
\end{equation}
In a well-specified model, this condition can be satisfied for all query-document pairs where $\pi(d, k \mid q) > 0$. Thus, the influence of the logging policy on the estimated parameters vanishes as long as identifiability is guaranteed.
\end{proof}

\begin{lemma}[Logging policy impact on misspecified models]
\label{lemma:misspecified}
When a two-tower model is misspecified, systematic correlations between the model's residual errors and the logging policy can introduce bias in parameter estimation, even when identifiability is guaranteed.
\end{lemma}

\begin{proof}
When the model is unable to match the true click probabilities, we have a non-zero residual click prediction error:
\begin{equation}
    \epsilon(q,d,k) \equiv P(C\!=\!1 \mid q,d,k) - \sigma(\theta_k+\gamma_{q,d}) \neq 0,
\end{equation}
with the gradient conditions becoming:
\begin{equation}\label{eq:gradient-residual}
    \sum_k\pi(d,k\mid q)\,\epsilon(q,d,k)=0,
     \ \
    \sum_q\sum_d\pi(d,k\mid q)\,\epsilon(q,d,k)=0.
\end{equation}
These conditions state that the policy-weighted averages of the residuals must vanish across positions and across items. However, when our model's click prediction errors $\epsilon(q,d,k)$ are systematically correlated with the policy's display probability $\pi(d, k \mid q)$, the optimizer must shift model parameters from their true values to satisfy these conditions. Note that model misspecification alone can bias model parameters as the average residual error might be non-zero even under a uniform random logging policy. But the logging policy can introduce additional bias in parameter estimation beyond what would be caused by model misspecification.
\end{proof}

\noindent%
Our analysis reveals a nuanced relationship between logging policies and additive two-tower models. While perfectly specified models remain unaffected by logging policies (beyond identifiability concerns), real-world models might inevitably contain some degree of misspecification. This creates a vulnerability where systematic correlations between model errors and logging policy behaviors can significantly distort parameter estimation. \emph{The more a logging policy systematically favors certain documents for certain positions, the more these correlations can amplify estimation bias.}

\vspace*{-1mm}
\subsection{Common sources of model misspecification}
In the following, we highlight three scenarios where model misspecification creates errors systematically correlated with logging policies, leading to biased parameter estimation. We will omit a detailed mathematical analysis of each use case for brevity.

\paragraph{Functional form mismatch} Model misspecification arises when our model does not match the true data-generating process. For example, we fit a linear relevance model to data following non-linear patterns. If this mismatch leads to incorrect click predictions for query-document pairs with certain features and our logging policy associates these features with specific positions, we have a problematic correlation between our residuals and our logging policy, leading to biased model parameters.

\paragraph{Omitted variable bias} A second common problem is a logging policy having access to more information than our current model. If our logging policy used features that are predictive of user preference to place certain query-document pairs in certain positions (e.g., through business rules), but our current model does not have access to these features, our residual prediction errors will be systematically correlated with position due to omitted variable bias~\cite{Wilms2021OmittedVariableBias}.

\paragraph{Expert policy in simulation}\label{sec:expert_policy} The third case is a subtle version of omitted variable bias, which is relevant in simulation experiments. In fact, it is the problem in our motivating example in Fig.~\ref{fig:example}.\footnote{We include results of our motivating example without an expert policy in Appendix~\ref{appendix:motivating-example}.}
A simple method to simulate a strong logging policy used by related work is sorting by the expert annotated labels provided in classic LTR datasets, i.e., sorting documents from most to least relevant~\cite{Zhang2023Disentangling,Deffayet2023CMIP}. However, note that the features provided in real-world LTR datasets~\cite{Qin2013MSLR,Chapelle2011Yahoo,Dato2022Istella22} cannot perfectly predict the expert relevance labels. Thus, sorting directly by these labels instead of sorting by a model trained on these datasets is essentially like a logging policy having access to a very predictive relevance feature that is being withheld from our current model, leading to omitted variable bias. In the following, we propose a simple method to reduce the effect of logging policies on misspecified two-tower models.

\section{RQ3: Sample Weights for Misspecified Models}
\label{sec:IPS}

While the focus of this work is analytical, our findings on logging policy effects suggest a natural fix for training two-tower models. The priority should be to fix model misspecification, as this lowers the impact of unequal exposure through the logging policy. However, should that not be possible (e.g., because business rules that impacted our past rankings were not documented), we can reduce the impact of uneven document distributions by weighting each query-document pair inversely to its propensity of being displayed:
\begin{equation}
\mbox{}\hspace*{-2mm}
    \begin{split}
        \mathcal{\hat{L}_{\text{IPS}}}(\theta, \gamma: D) = - \frac{1}{N} \sum_{(q,d,k,c) \in \mathcal{D}} \frac{1}{\pi(d,k\mid q)} \bigl[c\log \sigma(\theta_k + \gamma_{q,d}) \\
        + (1-c)\log (1-\sigma(\theta_k + \gamma_{q,d}))\bigr].
    \end{split}
\hspace*{-2mm}\mbox{}    
\end{equation}
Note how this approach is different from traditional inverse propensity scoring (IPS) methods used in unbiased learning-to-rank~\cite{Joachims2017IPW,Oosterhuis2020PolicyAware} as we target different biases: existing pointwise IPS methods primarily correct for position bias~\cite{Bekker2019PointwiseIPS,Saito2020PointwiseIPS,Hager2023ClickModelIPS}. In contrast, this sample weight corrects for the uneven document distribution across ranks. The loss is essentially fitting a model against a version of the dataset in which all documents appeared equally in all positions, which is related to corrections used in the position bias estimation literature~\cite{Agarwal2019AllPairs,Fang2019InterventionHarvesting,Benedetto2023ContextualBias}, the policy-aware IPS estimator for selection bias~\cite{Oosterhuis2020PolicyAware,Li2018OffPolicyClickModels}, or the recent work by \citet{Luo2024UnbiasedPropensity}. Note that estimating this propensity is a challenge on its own. In this paper, we calculate the propensities by simply counting how many times a query-document pair was displayed in a given position.

\begin{figure*}[h!]
    \includegraphics[width=1\textwidth]{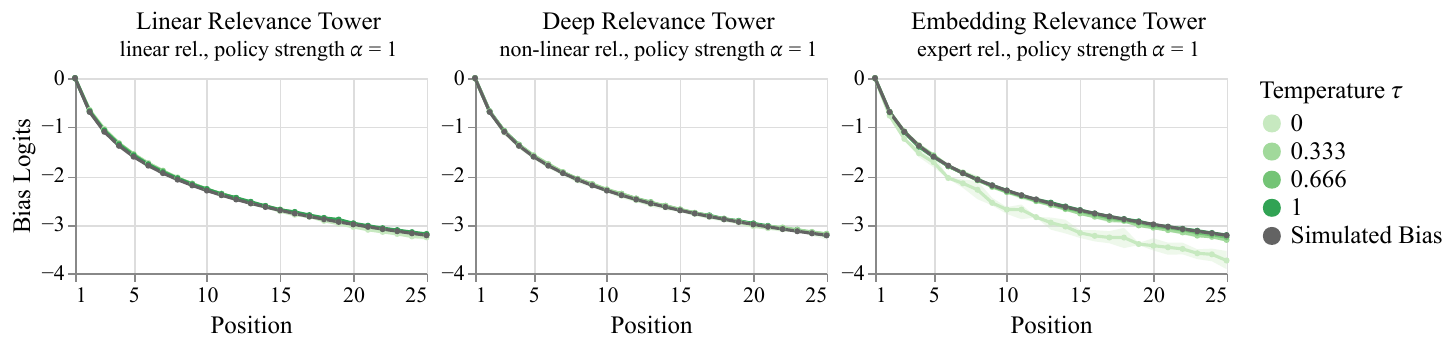}
    \caption{Evaluating position bias for three well-specified models matching the simulated user behavior. All models converge to ground-truth bias parameters while being trained on the strongest logging policy $(\alpha = 1)$. Feature-based models recover position bias without any document swaps ($\tau = 0$), while the embedding-based model requires swaps for identification ($\tau > 0$).}
    \label{fig:model-fit}
\end{figure*}

\vspace*{-1mm}
\section{Experimental Setup}
To evaluate our theoretical claims about model identifiability and logging policy effects, we implemented a comprehensive simulation framework. This controlled setting allows us to measure how well two-tower models can recover their underlying parameters under various conditions. Rather than solely relying on ranking metrics (such as nDCG), which can obscure underlying parameter estimation issues~\cite{Deffayet2023Robustness,Deffayet2023CMIP}, we directly compare inferred bias parameters against their ground-truth values to assess model identifiability. This provides clearer insights, as correctly estimated position bias parameters indicate properly identified relevance parameters due to their complementary relationship. By controlling logging policy strength, position randomization levels, and model specification, we will test each component of our theoretical analysis in isolation.

\paragraph{Datasets} The basis for our simulations are traditional learning-to-rank datasets with query-document features with an expert-annotated relevance label between 0-4. While we considered multiple datasets~\cite{Chapelle2011Yahoo,Dato2016Istella,Dato2022Istella22}, we will focus the discussion in the rest of the paper solely on MSLR30K~\cite{Qin2013MSLR} from the Bing search engine (31,531 queries and 136 dimensional feature vectors). We make this choice as our experiments consider very idealized simulations in which the choice of base dataset barely matters. We use the official training, validation, and testing splits. To increase computational efficiency, we truncate the number of documents per query to the 25 most relevant documents and drop all queries without a single relevant document ($\approx 3.1\%$ of queries). These steps purely aid the scalability of our simulation and all of our findings hold up on the full dataset. Next, we normalize each query-document feature provided using $\text{log1p}(x)=\ln(1 + |x|) \odot \text{sign}(x)$ following \citet{Qin2021AreNeuralRankers}.

\paragraph{Synthetic relevance} \label{sec:synthetic-relevance} LTR datasets come with relevance labels obtained by human experts~\cite{Qin2013MSLR,Chapelle2011Yahoo}. To isolate the effect of model mismatch, we additionally generate synthetic relevance labels that follow a known model class based on the provided query-document features. We generate a linear relevance label with Gaussian noise:
\begin{equation}\label{eq:linear-relevance}
    \hat{\gamma}_{q,d} = w^T x_{q,d} + \xi_{q,d}, \quad \xi_{q,d} \sim \mathcal{N}(0, \sigma^2),
\end{equation}
and a non-linear relevance label using a two-layer neural network with 16 neurons and tanh activations:
\begin{equation}\label{eq:non-linear-relevance}
    \hat{\gamma}_{q,d} = W_2^T \text{tanh}(W_1^T x_{q,d} + b_1) + b_2 + \xi_{q,d}, \quad \xi_{q,d} \sim \mathcal{N}(0, \sigma^2).
\end{equation}
In all our experiments we initialize the model weights randomly and add Gaussian noise with standard deviation $\sigma = 0.2$ to each label. Finally, we apply min-max scaling to ensure the resulting labels are in a comparable range to expert annotations, mapping the 5th to 95th percentiles of the synthetic labels to the range [0,4].

\paragraph{Logging policy} We isolate ranking performance and positional variability in our logging policy. To create a strong pointwise ranker, we train a relevance tower on the complete ground-truth relevance annotations of the training set with a mean-squared error loss. Note that we intentionally train our logging policy on the full training dataset and do not adhere to the common ULTR practice of training a weak ranker on 1\% of the dataset~\cite{Joachims2017IPW,Ai2018DLA,Jagerman2019ModelIntervene,Oosterhuis2020PolicyAware} as we need the strongest policy that our pointwise two-tower models could still capture. Our final logging policy score $s_{q,d}$ interpolates between our logging policy predictions and random noise:
\begin{equation}
\mbox{}\hspace*{-2mm}
    s_{q,d} = \operatorname{sign}(\alpha) \left(|\alpha| \cdot \hat{\gamma}_{q,d} + (1 -  |\alpha|) \cdot u_{q,d} \right), \ u_{q,d} \sim U(0, 4),
\hspace*{-2mm}\mbox{}    
\end{equation}
where $\alpha \in [-1, 1]$ is a hyperparameter to set the logging policy \emph{strength}. We simulate three levels of policy strength, using $\alpha = 1$ for the best possible ranking, $\alpha = 0$ for a random ranking, and $\alpha = -1$ for an inverse ranking that is worse than random by ranking documents from least to most relevant. Note that we sample noise only once per query-document pair and otherwise keep the ranking fixed across all user sessions for the same query. The noise interpolation setup follows \citet{Zhang2023Disentangling} but does not introduce additional biases as we sort by a trained logging policy and not directly by ground-truth expert labels (see Section~\ref{sec:expert_policy}).

Next, we transform our deterministic logging policy into a stochastic policy to introduce varying levels of positional variability into our simulation process, a crucial component to create overlap for our identifiability experiments. While various approaches exist to create stochastic policies, including the Gumbel-Max trick~\cite{Maddison2017Concrete,Bruch2019Softmax} or randomization schemes from the position bias literature \cite{Joachims2017IPW,Radlinski2006FairPairs}, we opt for an epsilon greedy strategy for simplicity and interpretability~\cite{Watkins1989EGreedy}. With this approach, we show a uniform random ranking with probability $\tau$, and otherwise rank deterministically according to our logging policy scores $s_{q,d}$. We refer to this probability $\tau$ as the \emph{temperature} of our logging policy, with $\tau = 0$ indicating a deterministic policy (showing the same ranking across all sessions) and $\tau = 1$ indicating a uniform random ranking per session.

\paragraph{Click simulation} We simulate user clicks with position bias following the two-tower model and leave an investigation of user model mismatch to future work. We define position bias logits as: $\hat{\theta}_k = - \ln(k)$ where $k$ is the rank of the current query-document pair. The final click probability is: $P(C = 1 \mid q, d, k) = \sigma(\hat{\theta}_k + (\hat{\gamma}_{q,d} - 2))$, with our final relevance logits being a zero-centered version of our relevance labels obtained earlier. We center the scores from their original range between [0, 4] to [-2, 2] to avoid click logits that might overly saturate the sigmoid. For all experiments, we simulate 1M user sessions for training and 500K sessions for validation and testing respectively. We repeat all simulations over three random seeds and plot the 95\% confidence interval across all figures.

\paragraph{Model} Lastly, we describe our model architecture. Our bias tower uses a single parameter per rank as the position bias logit $\theta_k$. Our relevance tower is either a single embedding parameter for each query-document pair $\gamma_{q,d}$, a linear model, or a two-layer feed-forward neural network (2 layers, 32 neurons, ELU activations). All models were optimized using AdamW~\cite{Loshchilov2019AdamW} with a learning rate of 0.003 and a weight decay of 0.01 up to 50 epochs, stopping early after three epochs of no improvement of the validation loss. The entire setup was implemented in Jax~\cite{Bradbury2018Jax}, Flax (NNX)~\cite{Heek2020Flax}, and Rax~\cite{Jagerman2022Rax}. Our code, data, and results are available here: \url{https://github.com/philipphager/two-tower-confounding}.

\begin{figure*}[h!]
    \includegraphics[width=1\textwidth]{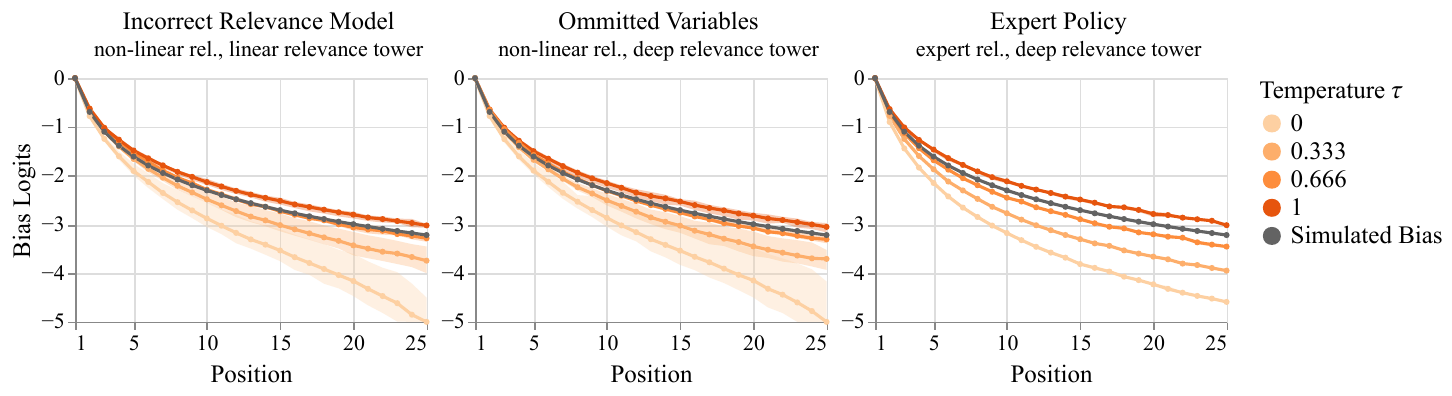}
    \caption{Evaluating position bias under model misspecification: (left) training a linear relevance model on non-linear user behavior, (middle) using fewer features when training two-tower models than are available to the logging policy, and (right) simulating an oracle logging policy by sorting by expert relevance labels from MSLR30K. The logging policy amplifies bias in all models. Model misspecification is visible as estimations do not match the simulated bias on fully randomized data ($\tau = 1$).}
    \label{fig:model-misfit}
\end{figure*}

\begin{figure*}[h!]
    \includegraphics[width=1\textwidth]{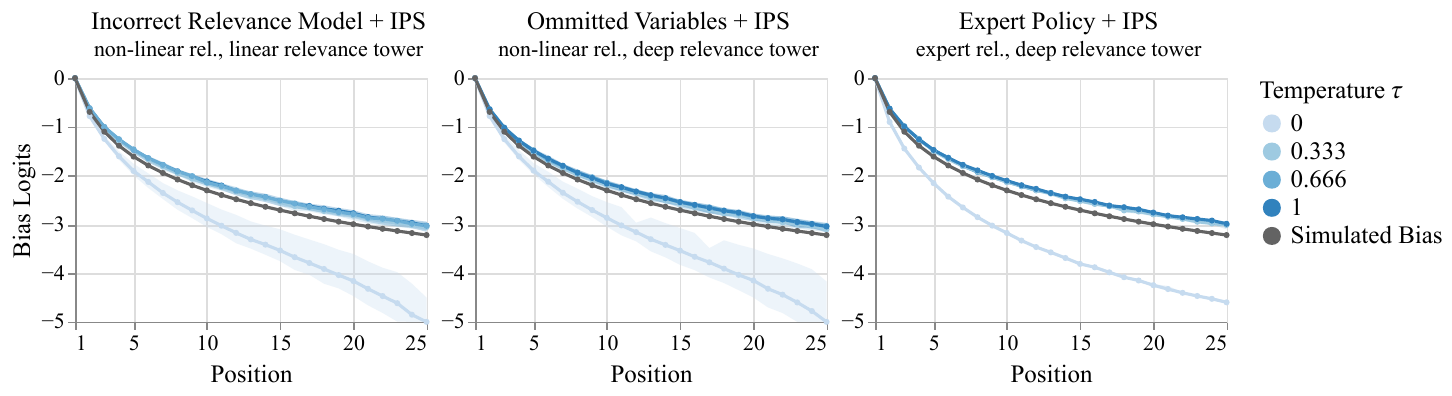}
    \caption{Reweighting samples during training inversely to the propensity of a policy placing an item in a given position. IPS requires that documents have a non-zero chance of occurring in other rank. Therefore, we see a helpful impact of IPS across various levels of randomization, but no improvement on a deterministic policy ($\tau = 0$). Note that while the logging policy impact is minimized, a bias remains in all three examples due to model misspecification.}
    \label{fig:model-misfit-ips}
\end{figure*}

\vspace*{-1mm}
\section{Results}
In the following, we empirically test the claims of our work. First, we will examine our theory of identification and the claim that logging policies do not impact well-specified two-tower models, meaning when our model assumptions match the actual user behavior. Second, we implement the three settings described in Section~\ref{sec:expert_policy} to investigate if logging policies amplify biases under model mismatch. Lastly, we evaluate whether the sample re-weighting scheme described in Section~\ref{sec:IPS} can alleviate the logging policy impact in each of the three model mismatch settings. Due to spacial constrains, we will only visualize models trained on the strongest logging policies (with varying temperature parameters), as this is the setting for which we expect the largest logging policy effects. Full evaluation results are available in our online repository.

\vspace*{-2mm}
\subsection{Evaluating well-specified models}
First, we evaluate our theory of identifiability stating that two-tower models with separate relevance parameters per query-document pair require exact document swaps. We display position bias estimated by a two-tower model trained with separate query-document parameters in the right panel of Fig.~\ref{fig:model-fit}. The model fails to infer the simulated bias under a deterministic policy ($\tau = 0$) when no document swaps are available. All other temperature settings introducing swaps across positions lead to an identified model.

The second part of our identifiability theory states that two-tower models can be identified from query-document features, as long as there are overlapping feature distributions between ranks. The left panel in Fig.~\ref{fig:model-fit} displays a linear relevance tower trained on clicks following a synthetic linear relevance model, and the middle panel a model using a neural relevance tower trained on users following a non-linear relevance distribution (both generated as described in Section~\ref{sec:synthetic-relevance}). The results confirm that under idealized conditions, model parameters can be perfectly recovered even on a fully deterministic policy, showcasing that features can fully identify two-tower models.

The last observation we can make in Fig.~\ref{fig:model-fit} is related to our claim that logging policies have no effect on well-specified two-tower models that capture the true data generating process. Once identifiability is guaranteed (through feature overlap or explicit document swaps), all three models are completely unaffected by different levels of simulated randomization that coincide with different logging policy performance. Our first experiment confirms our theory on identification as well as our claim that logging policies do not affect well-specified two-tower models once identifiability is guaranteed.

\vspace*{-2mm}
\subsection{Evaluating misspecified models}
Next, we investigate three scenarios of systematic model misspecification which, based on our theory in Section~\ref{sec:logging_policy}, should introduce a bias that should be amplified by our logging policy. Fig.~\ref{fig:model-misfit} displays the resulting position bias parameters. In the left panel, we simulate fitting an incorrectly specified relevance model by training a linear relevance tower on data following a non-linear relevance behavior. In the middle panel, we simulate omitted variable bias by training a logging policy on all 136 query-document features of MSLR30K but make only the first 100 features available to the downstream two-tower model. And the right panel creates model misfit by using the simulation setup of \citet{Deffayet2023CMIP} and \citet{Zhang2023Disentangling}, who rank query-document pairs based on expert-annotated relevance labels to create a strong logging policy. In all three cases, we first observe a notable bias in the estimated parameters that coincides with logging policy correlation (a higher temperature implies more randomness and less correlation). Secondly, we observe that a lower temperature (which creates a more unequal query-document distribution across ranks) amplifies bias. In other words, as rankings deviate more from uniformly random exposure, the magnitude of bias increases. Importantly, this bias manifests just as strongly when simulating an inverse logging policy that deliberately positions all relevant documents at the bottom (Appendix~\ref{appendix:extended-misspecification}).
This confirms that non-uniform query-document exposure across ranks can intensify biases in misspecified two-tower models. Lastly, we can also see that a uniform random ranking ($\tau = 1$) gets close but does not perfectly match the simulated bias, even under these highly idealized conditions. This confirms that model misfit on its own introduces bias, independent of logging policy effects.

\vspace*{-2mm}
\subsection{Addressing unequal exposure using IPS}
Our last experiment tests if the sample weighting scheme described in Section~\ref{sec:IPS} can help alleviate the effects of unequal document exposure on misspecified models. Fig.~\ref{fig:model-misfit-ips} displays the results of applying the weighting scheme to the same three simulation setups used in our previous model mismatch experiments. We can see that the weighting scheme does not work under a deterministic logging policy $\tau = 0$, as documents have a probability of one to be displayed in their initial rank and a propensity of zero to be displayed in any other rank. Propensity scoring requires positivity~\cite{DAmour2021Overlap}, meaning in our case that documents have a non-zero chance to be displayed at other ranks. We can see that in all other cases when we introduce increasing levels of variability into our policy, the models converge to identical position bias estimates, independent of the logging policy. Note that the estimated parameters are still slightly biased due to model misspecification, however, the bias amplification by the logging policy has been dampened. While this setting is highly idealized, it confirms that weighting query-document pairs inversely to their propensity of being displayed at a given rank, can help counterbalance unequal exposure effects in two-tower models.

\vspace*{-1mm}
\section{Discussion}
In the following, we discuss the practical implications of our findings and the limitations of our work.

\vspace*{-2mm}
\subsection{Practical implications}
\label{sec:practical_implications}

Our theory offers several important implications for those working with two-tower models. First, we find that models generalizing over query-document features require overlapping feature distributions across ranks to ensure identifiability. Second, while logging policies do not impact well-specified two-tower models, they can significantly amplify bias in models that systematically fail to capture the actual user behavior. Therefore, we advise determining the underlying user behavior model (e.g., using randomized experiments) and continuously monitoring click prediction residuals for correlations with the logging policy to detect emerging mismatches between model assumptions and reality. Accurately modeling users might create tensions: practitioners may be tempted to condition on more features to improve click predictions and minimize model misfit, but increasing feature dimensionality reduces overlap between ranks, potentially compromising model identifiability. Similarly, reducing the expressivity of model towers (e.g., using shallow networks and regularization) can improve identifiability but risks underfitting the data, leading to misspecification bias that logging policies may amplify. 

Collecting randomized data helps address these issues by creating overlap and guaranteeing identifiability while reducing unequal document exposure across ranks, thus lessening the impact of model misspecification. However, this must be balanced against a potential disruption of the user experience. For researchers working in simulation, we caution against creating logging policies by directly sorting documents according to expert annotations. This approach introduces a subtle form of omitted variable bias, simulating a logging policy with access to perfect relevance information that the two-tower model cannot capture. Instead, researchers should train a separate logging policy model from features. Finally, when some degree of model misspecification is unavoidable, consider propensity weighting to counteract potential bias amplification through the logging policy.

\vspace*{-2mm}
\subsection{Limitations}
\label{sec:limitations}

Our work has several limitations that motivate directions for future research. First, we focus on identifiability, a concept fundamentally evaluated in idealized conditions of a well-specified model and infinite data~\cite{Lewbel2019IdentificationZoo}. While we confirmed our theories empirically, our simulations represent highly idealized settings tied specifically to additive two-tower models. Future work might develop empirical tests for the identifiability of models on specific real-world datasets (practical identifiability~\cite{Raue2009ProfileLikelihood}). Note that even when data collection and model assumptions match in theory, empirical datasets might be too small (e.g., containing too few swaps across positions) to practically identify a two-tower model.
Second, we focused on a single prevalent user model. Future research could explore how various user models interact with logging policies or consider empirical tests for identifiability that allow plugging in other models with different assumptions.
Third, while we show that similar query-document features across ranks can identify two-tower models, future work might develop diagnostic metrics to quantify the overlap between high-dimensional query-document feature distributions on a given dataset. Promising directions to measure overlap include distance measures~\cite{Kolouri2019SlicedWasserstein,Gretton2012MMD}, hypothesis testing~\cite{LopezPaz2019ClassifierTwoSample,Gretton2012MMD}, or positivity measures from causal inference~\cite{DAmour2021Overlap}.
Lastly, our propensity scoring method mainly demonstrates that we identified the correct mechanism. To be practically useful, future research might consider better propensity estimation techniques and variance reduction methods.

\vspace*{-1mm}
\section{Conclusion}
We set out to explain why we observe degrading ranking performance in two-tower models on data collected by strong logging policies. To investigate this phenomenon, we first analyzed under which conditions we can identify the parameters of additive two-tower models from observational data and the impact of logging policies on parameter estimation. 

We first show that two-tower models can be identified from rankings without document swaps as long as feature distributions across ranks overlap. Secondly, we find that logging policies have no impact on well-specified two-tower models, meaning when our model assumptions match the actual user behavior.

However, misspecified two-tower models making systematic click prediction errors lead to biased parameter estimates. If these prediction errors also correlate with document placement across positions, we find that logging policies can amplify this problem and lead to stronger bias in model parameters, which we show in three exemplary use cases. 

Lastly, we show that weighting each query-document pair by their probability of being displayed to the user can help mitigate additional bias introduced by the logging policy. However, the method cannot fully resolve biases introduced by model misspecification.

\vspace*{-1mm}
\begin{acks}
We thank Philip Boeken and Shashank Gupta for their insightful discussions and valuable feedback. This research was supported by the Mercury Machine Learning Lab created by TU Delft, the University of Amsterdam, and Booking.com.
Maarten de Rijke was supported by the Dutch Research Council (NWO), project nrs 024.004.022, NWA.1389.20.\-183, and KICH3.LTP.20.006, and the European Union's Horizon Europe program under grant agreement No 101070212. All content represents the opinion of the authors, which their employers or sponsors do not necessarily endorse.
\end{acks}

\clearpage
\bibliographystyle{ACM-Reference-Format}
\balance
\bibliography{main}

\clearpage
\appendix

\section{Bias features beyond position}
\label{appendix:beyond-position}

Two-tower models owe their popularity in industry applications partially to their ability to use bias features beyond position that might impact an item's examination probability~\cite{Chen2023-TENCENT-ULTR-1,Yan2022TwoTowers}. These features might include, a.o., device platform (mobile vs. desktop), content type (text vs. video), or the display height of the document on screen~\cite{Zou2022Baidu}. Such information can be captured in a bias feature vector $z_{q,d} \in \reals^n$ serving as input to a dedicated bias tower $b(\cdot)$, commonly implemented a second neural network:
\begin{equation}
    P(C = 1 \mid q,d) = \sigma(b(z_{q,d}) + r(x_{q,d})).
\end{equation}
In this expanded model, all discussed identifiability principles from Section~\ref{sec:Identification-feature-overlap} still apply but with increased combinatorial complexity. The model remains identifiable only when sufficient overlap exists in document distributions across the entire space of bias features:
\begin{equation}
    \begin{split}
     \supp{P(x \mid z)} &\cap \supp{P(x \mid z')} \neq \emptyset.
    \end{split}
 \end{equation}
This requirement presents a combinatorial challenge, as the number of potential bias configurations grows exponentially with each additional bias dimension.

To illustrate, consider a setting with five positions, three content types, and two device types. This creates 30 distinct bias configurations that must form a connected graph through overlap in the query-document feature distribution. Without careful data collection, certain bias combinations may remain isolated, leading to unidentifiable parameters in those regions of the feature space. In practice, this may require explicit randomization across bias dimensions or the collection of vast datasets to capture enough natural variation across all bias configurations.

\section{Motivating example without expert policy}
\label{appendix:motivating-example}

We revisit our motivating example from the introduction and replace the expert policy (sorting by ground truth annotations) with a deep neural network trained on expert annotations. Removing this problematic source of model misspecification from our simulation greatly reduces the observed  logging policy impact.

\begin{figure}[H]
    \hspace*{-2em}
    \includegraphics[clip, trim=0mm -18mm 0mm 0mm, width=0.9\linewidth]{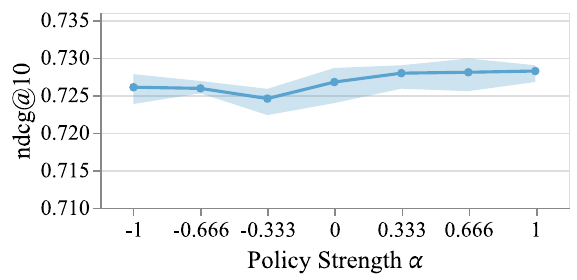}
    \vspace*{-1.4cm}
    \caption{Two-tower models trained on deterministic logging
    policies of varying strengths $\alpha$. The logging policy is a deep relevance tower trained on expert annotations instead of directly sorting by ground-truth relevance labels, as done in the introduction. This showcases the importance of model misspecification in ranking performance.}
    \label{fig:example-deep-policy}
\end{figure}

\section{Extended simulation results on model misspecification}
\label{appendix:extended-misspecification}

In the following, we provide more visualizations on our three cases of model misspecification. First, we display position bias logits estimated on the inverse policy in Fig.~\ref{fig:model-misfit-inverse-lp}, i.e., when placing the least relevant items on top which is worse than random ranking and creates correlation between document placement and the model's prediction errors. Second, we display the corresponding ranking performance in terms of nDCG@10, measured against the respective relevance label used in each experiment in Fig.~\ref{fig:model-misfit-ranking}. 

\begin{figure}[H]
    \includegraphics[width=0.9\columnwidth]{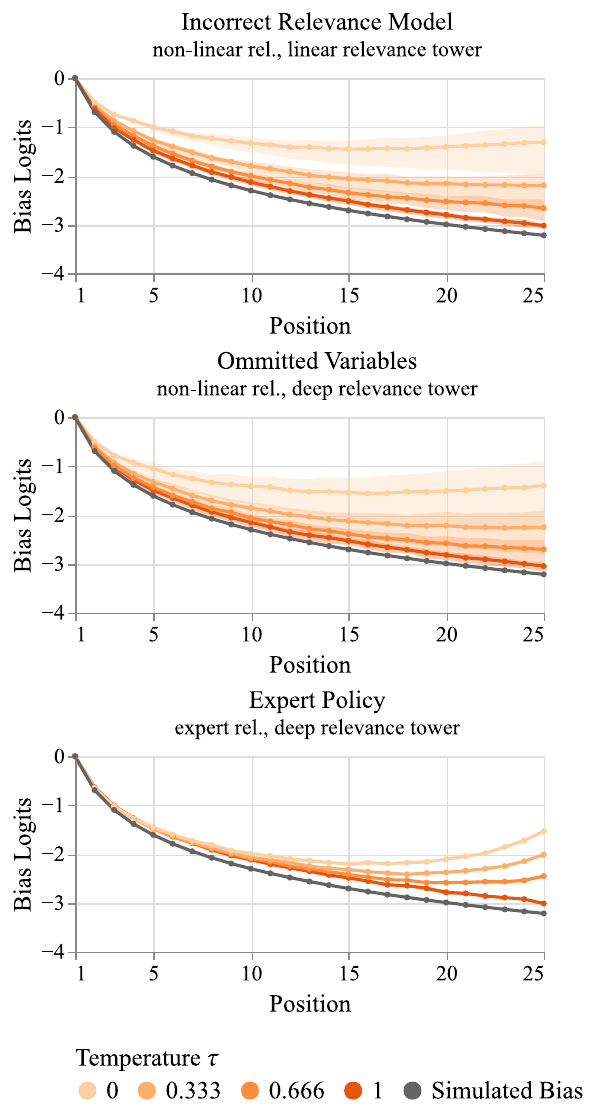}
    \caption{Position bias estimation under model misspecification on an inverse policy ($\alpha = -1$) placing the most relevant items at the bottom: (left) training a linear relevance model on non-linear user behavior, (middle) training the logging policy on more features than available to the two-tower model, and (right) sorting by expert relevance labels from MSLR30K as logging policy. The results confirm recent work on propensity over/underestimation by \citet{Luo2024UnbiasedPropensity}.}
    \label{fig:model-misfit-inverse-lp}
\end{figure}

\begin{figure}[H]
    \includegraphics[width=0.9\columnwidth]{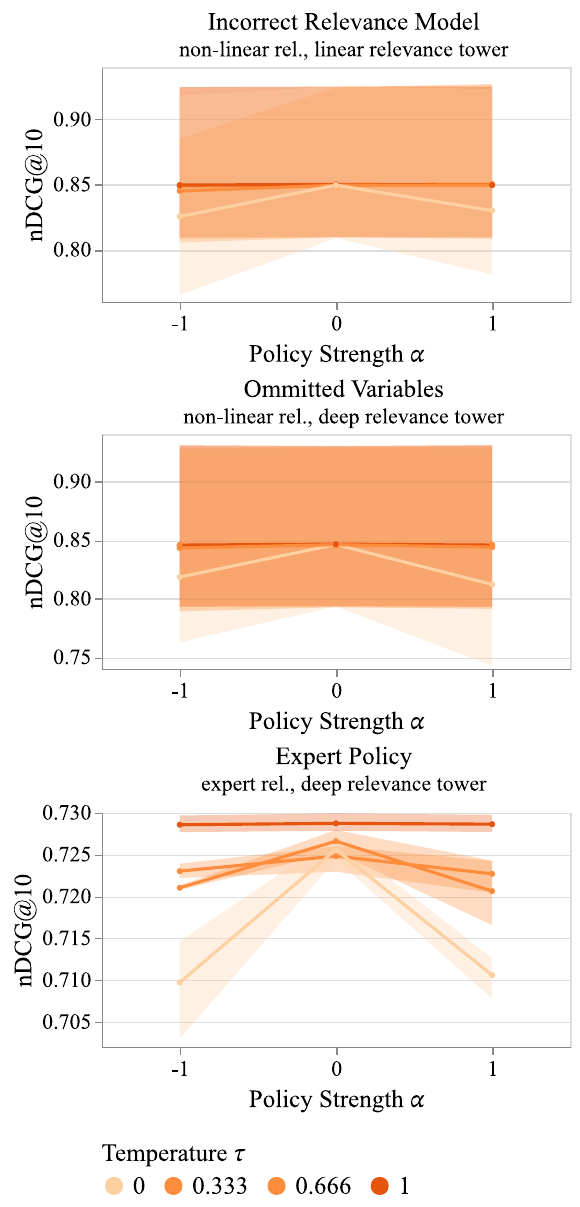}
    \caption{Ranking performance of the relevance tower under model mismatch across different logging policy strengths and temperatures. While the overall ranking performance is significantly lower than the ideal (an nDCG@10 of 1 is possible in the left and middle panel), ranking performance does not always correspond to the clear bias in position parameters highlighted in Fig.~\ref{fig:model-misfit} and Fig.~\ref{fig:model-misfit-inverse-lp}, at least not on our dataset size. This echos recent findings on the unreliability of evaluating click model parameters through nDCG~\cite{Deffayet2023Robustness,Deffayet2023CMIP}.}
    \label{fig:model-misfit-ranking}
\end{figure}

\end{document}